\documentclass[a4paper,11pt]{article}
\usepackage[utf8]{inputenc}

\usepackage{graphicx}
\usepackage{geometry}
\usepackage{amsthm}
\usepackage{amssymb}
\usepackage{amsmath}
\usepackage{hyperref}
\usepackage{xcolor}
\usepackage{enumerate}
\usepackage{listings}
\usepackage{complexity}
\usepackage{microtype}
\usepackage{todonotes}
\usepackage[inline]{enumitem}

\usepackage[font=small,labelfont=bf]{caption}
\usepackage[font=small,labelfont=normalfont,labelformat=simple]{subcaption}

\graphicspath{{./figures/}}

\newtheorem{theorem}{Theorem}
\newtheorem{definition}{Definition}
\newtheorem{lemma}[theorem]{Lemma}

\newtheorem{corollary}[theorem]{Corollary}

\DeclareMathOperator{\cf}{cf}

\def\inst#1{$^{#1}$}

\date{}

\title{On Crossing-Families in Planar Point Sets\footnote{
		Research for this article was initiated in the course of the bilateral research project ``Erd\H{o}s--Szekeres type questions for point sets'' 
		between Graz and Prague, supported by the OEAD project CZ~18/2015 and project no.\ 7AMB15A~T023 of the Ministry of Education of the Czech Republic.
		J.K. was also supported by the grant no. 21-32817S of the Czech Science Foundation (GA\v{C}R) and by Charles University project UNCE/SCI/004. 
		M.S.\ was supported by the DFG Grant SCHE~2214/1-1.
		P.V.~was supported by the grant no. 21-32817S of the Czech Science Foundation (GA\v{C}R).
		B.V.\, partially supported by Austrian Science Fund within the collaborative DACH project \emph{Arrangements and Drawings} as FWF project \mbox{I 3340-N35}. 
}}

\begin{document}
	
	\author{
		Oswin Aichholzer\inst{1}
		\and
		Jan Kyn\v{c}l\inst{2}
		\and
		Manfred Scheucher\inst{3}
		\and
		Birgit Vogtenhuber\inst{1}
		\and 
		Pavel Valtr\inst{2}
	}

	\maketitle

	\begin{center}
		{\footnotesize
			\inst{1} 
			Institute of Software Technology, Graz University of Technology, Austria \\
			\texttt{oaich@ist.tugraz.at}, \texttt{bvogt@ist.tugraz.at}
			\\\ \\
			\inst{2} 
			Department of Applied Mathematics, \\
			Faculty of Mathematics and Physics, Charles University, Czech Republic \\
			\texttt{kyncl@kam.mff.cuni.cz}
			\\\ \\
			\inst{3} 
			Institut f\"ur Mathematik, Technische Universit\"at Berlin, Germany\\
			\texttt{scheucher@math.tu-berlin.de}
			\\\ \\
		}
	\end{center}

\begin{abstract}
	A $k$-crossing family in a point set $S$ in general position is a set of $k$ segments spanned by points of $S$ such that all $k$ segments mutually cross.
	In this short note we present two statements on crossing families which are based on sets of small cardinality:
	(1)~Any set of at least 15 points contains a crossing family of size~4.
	(2)~There are sets of $n$ points which do not contain a crossing family of size larger than~$8\lceil \frac{n}{41} \rceil$.
	Both results improve the previously best known bounds.
\end{abstract}

%=================================================================================================

\section{Introduction}

Let $S$ be a set of $n$ points in the Euclidean plane  \emph{in general position}, 
that is, no three points in $S$ are collinear. 
A \emph{segment} of $S$ is a line segment with its two endpoints (which we will also call vertices) being points of~$S$. 
\begin{definition} 
	A \emph{$k$-crossing family} in a point set $S$ is a set of $k$ segments spanned by points of $S$ such that all $k$ segments mutually cross in their interior.
\end{definition} 

For a point set $S$, let $\cf(S)$ be the maximum size of a crossing family in~$S$, 
and let $\cf(n)$ be the minimum of $\cf(S)$ over all point sets $S$ of cardinality~$n$ in general position. 

\begin{figure}%[htb]
	\centering
	\includegraphics[page=2,scale=0.4]{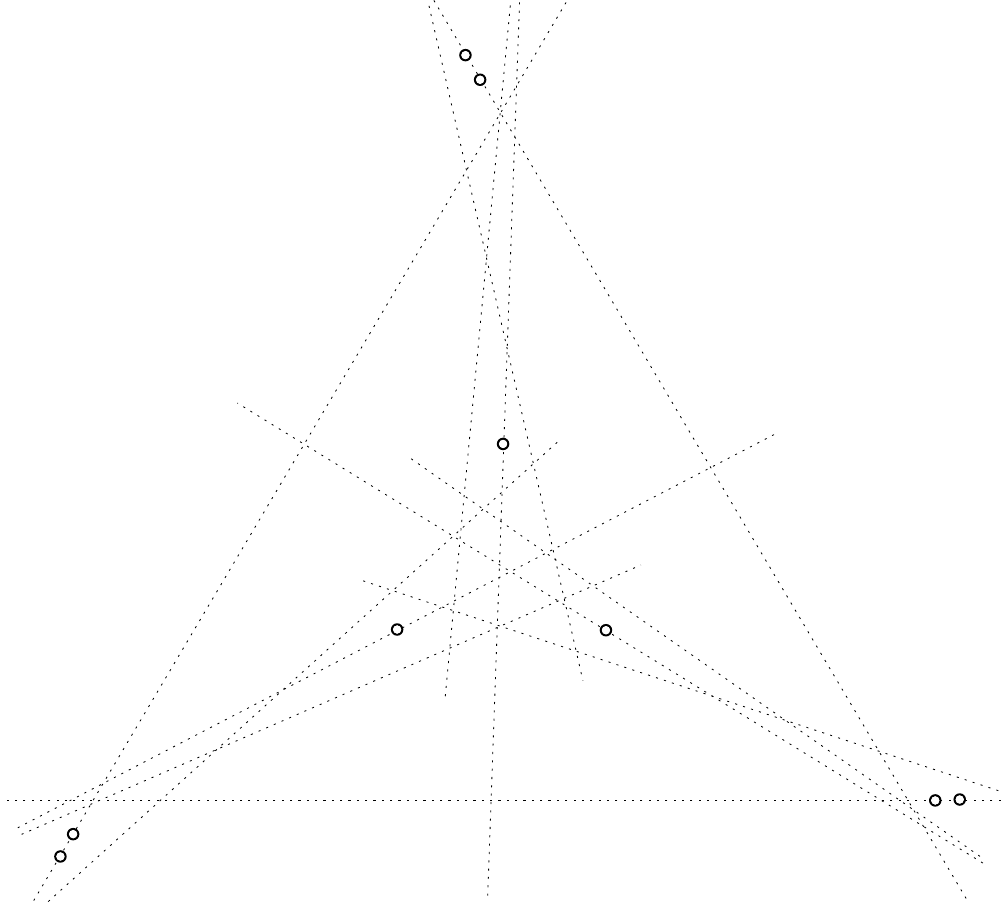} 
\caption{\label{fig:n9no3fam}
	A set $S$ of 9 points which does not contain a 3-crossing family, that is, $\cf(S)=2$.}
\end{figure}

  It is easy to see that $\cf(n)$ is a monotone function.
  From the fact that the complete graph with $5$ vertices is not
  planar it follows that any set of $n\ge 5$ points has a crossing
  family of size at least $2$. In~\cite{ak-psotd-01} it is shown that every set
  with $n \ge 10$ points admits a crossing family of size $3$. The
  result is based on analyzing the set of all order types of size~10. The bound on $n$ is tight, that is, there exist $12$ order
  types of $9$ points which do have a maximal crossing family of size
  $2$. One such set is shown in   Figure~\ref{fig:n9no3fam}.

 Aronov et al.~\cite{aegkkps-94} proved in 1994 the existence of crossing families of size
$\sqrt{n/12}$ for every set of $n$ points, which until recently was the best general lower bound. Their proof relies on the existence of a pair of mutually avoiding sets of size $\sqrt{n/12}$; two point sets are called \emph{mutually avoiding} if every line determined by a pair of points from one of the sets is disjoint from the convex hull of the other set. The lower bound on the size of mutually avoiding sets is asymptotically tight. 

Only in 2019 
Pach, Rubin, and Tardos~\cite{PRT2019} showed  in a breakthrough result that any set of $n$ points in general position in the plane contains a crossing family of size $n^{1-o(1)}$ (the full version of this paper appeared as~\cite{PACH2021107779}).
This almost shows the generally accepted conjecture that $\cf(n)$ should be in~$\Theta(n)$.
Further evidence for this conjecture comes from a result by Valtr~\cite[Theorem 14]{Valtr1996}, that a
set of points chosen independently at random from a convex shape contains with high probability a linearly sized crossing family.
Also the currently best upper bounds supported this conjecture. Recently, Evans and
Saeedi~\cite{evans2019problems} showed that $\cf(n) \le 5\lceil \frac{n}{24} \rceil$ on which we will improve.

\subsection{Results and motivation}

In Section~\ref{sec:15always4},
we use exhaustive abstract extension of order types and SAT solvers to investigate crossing families in small point configurations.
We verify some previous results and determine the value $\cf(15)=4$.
Based on our computational results, 
we conjecture that every set of 21 or more points contains a crossing family of size~5.

In Section~\ref{sec:upperbound},
we present a set of 41 points without 9-crossing families.
By utilizing this set, we improve the upper bound to $\cf(n) \le 8\lceil \frac{n}{41} \rceil$.

Our motivation to obtain these results comes, besides that we think that they are of their own interest, from a possible application of crossing families to other combinatorial enumeration problems.
For example, the existence of crossing families of a fixed size can be used to obtain results about the (asymptotic) number of geometric graphs of a certain class.
To obtain the first non-trivial lower bound for the number of triangulations of a set of $n$ points in the plane~\cite{ahn-lbntp-04} a recursive divide $\&$ conquer approach was used. To avoid over-counting, it was essential that the edges used to subdivide the point set are pairwise crossing, as no triangulation can contain two such edges at the same time. The induction base (how many mutually crossing edges on how few points can we guarantee) then determines the best possible base in the exponential bound. In~\cite{ahn-lbntp-04} this was optimized over crossing families of constant size up to $k=4$, resulting in a lower bound for the number of triangulations of $\Omega(2.33^n)$. Another example is related to the topic of the Geometric Optimization Challenge 2022~\cite{website_SoCG22}, which is part of the CG week, namely the problem of partitioning a given geometric graph into the minimum possible number of plane subgraphs.
%''Minimum Partition into Plane Subgraphs''. 
One of the reference papers there is~\cite{ahkklpsw-ppstp-17}. In that paper, a lower bound on the number of edge-disjoint plane spanning trees was shown by using asymptotic results on the size of crossing-families.

Both mentioned results have meanwhile been improved by using other techniques, see for example~\cite{aahpsv-ilbnt-16} and~\cite{BiniazGarcia2020}, respectively. But they testify that both, improved asymptotic bounds on the size of crossing families as well as improvements for the size of crossing families in small point sets, can be used to strengthen general (asymptotic) bounds. We expect that our results might be useful in a similar manner and may stimulate further results in that direction.

%=================================================================================================

\section{Sets of 15 points always contain a 4-crossing family}
\label{sec:15always4}

From the mentioned result $\cf(10)=3$~\cite{ak-psotd-01} we also know that any set of 11 points contains a 3-crossing family, and no such set can contain a crossing family of size more than 5 (as at least $2k$ points are needed for a $k$-crossing family).
The following table shows how the maximal size of a crossing family is distributed among all combinatorially different point sets of size 11. It was computed with the help of the database of all order types of size 11 obtained in~\cite{ak-aoten-06}.

\begin{table}[h]
\begin{center}
\begin{tabular}{r|r|r}
\hline
       $k$    & number of order types & percentage \\  
\hline
\hline
3 &  63 978 178 & 2.7~\%\\
4 & 1 783 117 647 & 76.4~\%\\
5 &  487 417 082 & 20.9~\%\\
\hline
total & 2 334 512 907& 100.0~\%\\
\hline
\end{tabular}
\end{center}
\caption{The number of realizable order types of size 11 with a maximum crossing family of size~$k$.
}
\label{tab:crossfam11}
\end{table}

To obtain the largest point set containing no crossing family of size 4, we made a complete abstract order type extension from $n=11$ to $n=15$.
The database of all realizable order types of cardinality $11$ contains 2 334 512 907 sets~\cite{ak-aoten-06}, of which 63~978~178 (about 2.7~\%) contain no 4-crossing family; see Table~\ref{tab:crossfam11}. 
Since adding points to an existing set can never decrease the size of the maximal crossing family, we need to consider only those sets in order to iteratively find the largest set that contains no crossing family of size 4.

The approach of extending order types in an abstract way as described in~\cite{ak-aoten-06} has the advantage that there is no need to realize the obtained sets, as we actually are interested in the smallest cardinality where no such sets exist. This avoids dealing with the notoriously hard problem of realizing abstract order types, which is known to be $\exists \mathbb{R}$-hard~\cite{mnev1985manifolds}. The extension is done iteratively by adding one more (abstract) point in each step. Afterwards each obtained abstract order type is checked for the maximum size of a crossing family, and if this is at least $4$ the abstract order type is discarded.

After the first three rounds we obtained 2 727 858 abstract order types of cardinality 14 which do not contain a 4-crossing family. All these sets were extended by one further point (in the abstract setting), but all resulting sets contained a 4-crossing family.
Thus, the largest possible set with no 4-crossing family has size at most 14. 
The whole process of abstract extension took about 100 hours, computed in parallel on 40 standard CPUs.

To show that the obtained bound is best possible it is sufficient to realize at least one generated abstract order type of size 14, and such an example is given in Figure~\ref{fig:n14no4fam}. 
Together with the set of 20 points with no 5-crossing family depicted in Figure~\ref{fig:n20no5fam}, we conclude the following.

\begin{theorem}\label{thm:4crossingfamily}
Every set of at least $15$ points in the plane in general position contains a $4$-crossing family. 
Moreover, we have 
\[
\cf(n)=3 \text{ for } 10 \le n \le 14 
\quad
\text{ and }
\quad 
\cf(n)=4 \text{ for } 15 \le n \le 20.
\]
\end{theorem}

\begin{figure}%[htb]
	\centering
	\includegraphics{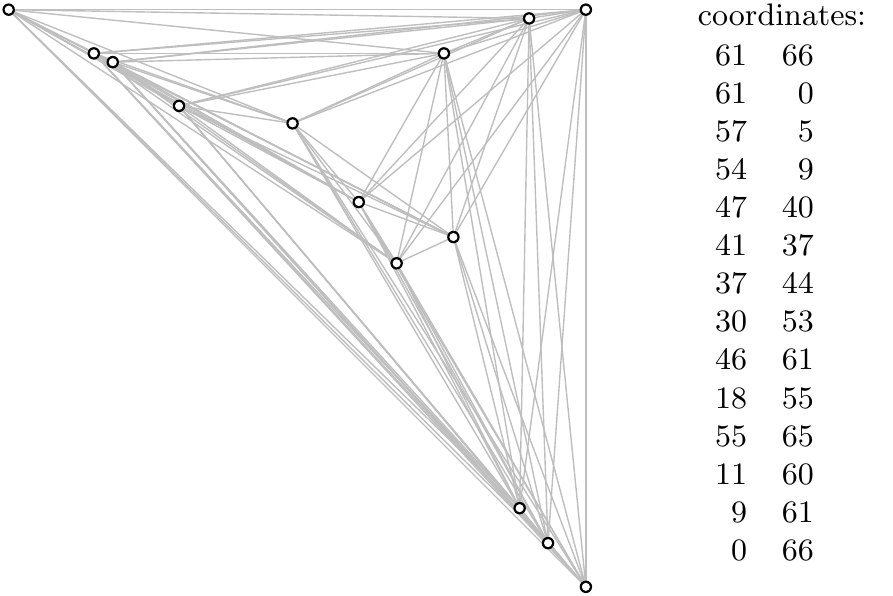}
	\caption{\label{fig:n14no4fam}
		A set $S$ of 14 points containing a 3-crossing family but no 4-crossing family, that is, $\cf(S)=3$.
	}
\end{figure}

Besides the above described computer proof,
we have also developed a SAT framework which allowed us to verify $\cf(15)>3$ within less than 2 CPU days using the SAT solver CaDiCaL~\cite{Biere2019}. 
While the instance is about 33~MB, 
the proof generated by CaDiCaL is 21~GB. 
The correctness of the proof can be verified 
using DRAT-trim \cite{WetzlerHeuleHunt2014} within 4 additional CPU days.

We have also used this framework to verify $\cf(10)>2$ \cite{ak-psotd-01}.
The python program creating the instance is available on our supplemental website \cite{website_crossing_families}.
The instance to decide whether or not $\cf(21)= 5$ is about 26~GB, but the solver did not terminate so far.

The idea behind the SAT model is very similar as in \cite{Scheucher2020}:
We assume towards a contradiction that $\cf(15) \le 3$, that is, there is a set of 15 points with no 4-crossing family.
We have Boolean variables $X_{abc}$ to indicate whether three points $a,b,c$ are positively or negatively oriented.
As outlined in \cite{Scheucher2020}, these variables have to fulfill the signotope axioms \cite{FelsnerWeil2001,BalkoFulekKyncl2015}.
Based on the variables for triple orientations, we then assign auxiliary variables $Y_{ab,cd}$ to indicate whether the two segments $ab,cd$ cross.
Finally we assert that for any set of four segments with pairwise 
distinct endpoints, at least one pair of them does not cross. 
As the SAT solver CaDiCaL terminates with ``unsatisfiable'',
no such point set exists, and hence $\cf(15) > 3$.

%=================================================================================================
\section{Small sets and an upper bound}
\label {sec:upperbound}

The following theorem was already implicitly used by Aronov et al.~\cite[Section~6]{aegkkps-94} in their discussion, 
where they stated the upper bound $\cf(n)\le\frac{n}{4}$. 
It can also be found as Lemma~3 in~\cite{evans2019problems}. 
Since in \cite{aegkkps-94} no proof is given and the proof from~\cite{evans2019problems} appears to be incomplete\footnote{In the proof in~\cite{evans2019problems}, each point is repalced by ``imperceptibly perturbed copies'' with no prescribed structure. However, this does not guarantee that odd cycles as illustrated in Figure~\ref{fig:wrongdoubling} are avoided.}, 
we include a full proof here. 

\begin{figure}[htb]
	\centering
	\includegraphics[page=2,scale=0.8]{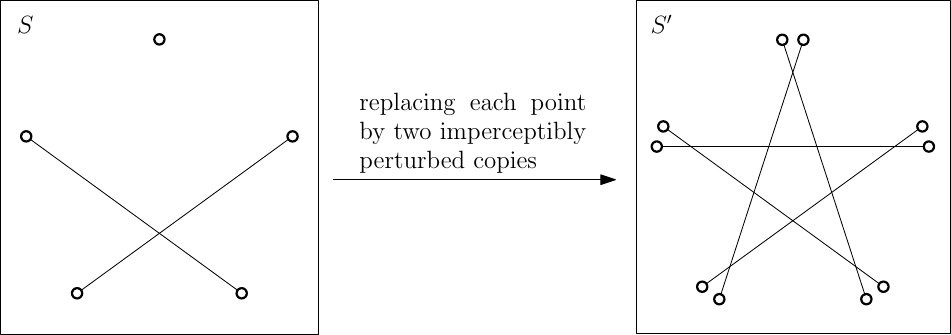} 
\caption{\label{fig:wrongdoubling}
	Five points in convex position have crossing family of size at most $2$ (left). Replacing each point by two imperceptibly perturbed copies with no prescribed structure as in~\cite{evans2019problems} may result in a set of 10 points with a crossing family of size larger than 4 (right). The process needs to avoid generating odd cycles of pairwise crossing edges; see the proof of Theorem~\ref{thm:factor} for details.}
\end{figure}

\begin{theorem}\label{thm:factor}
Let $S\subset \mathbb{R}^2$ be a set of $n$ points in general position with $\cf(S)=k$.
Then for any $N \ge n$ there exists a set $S'\subset \mathbb{R}^2$ of $N$ points in general position with $\cf(S') \le k\lceil \frac{N}{n} \rceil$.
\end{theorem}

For our proof of Theorem~\ref{thm:factor}, we will use a simple property of geometric thrackles.
A~\emph{geometric thrackle} is a geometric graph such that each pair of edges (drawn as line segments) either meets at a common vertex or crosses properly.
Woodall proved that a graph can be drawn as a geometric thrackle if and only if it is a a subgraph of a graph obtained by attaching leaves (vertices of degree $1$) to the vertices of an odd cycle~\cite[Theorem 2]{Woo71_deadlock}. We will need only the following weaker characterization.

\begin{lemma}
	\label{lem:even_cycles_thrackle}
	A geometric thrackle $T$ contains no even cycles.
\end{lemma}

A strenghtening to monotone thrackles was proved by Pach and Sterling~\cite{PS11_monotone_thrackles}. 
Since the proof for geometric thrackles is substantially simpler, we include it here to keep this note self-contained.

\begin{proof}[Proof of Lemma~\ref{lem:even_cycles_thrackle}]
	Assume there exists an even cycle $C=p_0,p_1,\ldots,p_n$ for some even $n \ge 4$ 
	with $p_0=p_n$ and $\overline{p_{i-1}p_{i}} \in T$ for $i=1,\ldots,n$.
	We set $p_i=p_j$ for $i=j\ (\text{mod } n)$.
	Consider a line segment $\ell=\overline{p_ip_{i+1}}$ in $C$. 
	The supporting line through $\ell$ divides the plane into two half-planes. 
	Since $T$ is a thrackle, 
	the previous segment $\overline{p_{i-1}p_{i}}$ 
	and the next segment $\overline{p_{i+1}p_{i+2}}$ cross, 
	and thus $p_{i-1}$ and $p_{i+2}$ lie in the same half-plane.
	Moreover, since all segments in the path $P=p_{i+2},p_{i+3},\ldots,p_{i-1}$ cross the segment $\overline{p_ip_{i+1}}$, 
	$P$ is an alternating path with respect to the side of the half-plane, 
	and hence $P$ has even length $|P|$, that is, an even number of edges.
	Since the cycle $C$ has length $|C|=|P|+3$ this is a contradiction, 
	because $C$ was assumed to have even length.
\end{proof}

\begin{proof}[Proof of Theorem~\ref{thm:factor}] 
	Let $S$ be a set of $n$ points in general position in the plane and let $m=\lceil \frac{N}{n} \rceil$. Without loss of generality we may assume that $N=mn$; in case $N<mn$ we may later remove some points from the constructed point set.
	We can also assume that all points of $S$ have distinct $x$- and $y$-coordinates; otherwise we slightly rotate~$S$. 
	Our aim is to construct a set $S'$ of $mn$ points by creating $m$ \emph{copies} of each point from $S$, such that the following two properties hold: 
	\begin{enumerate}[label=(S\arabic*), align=left, leftmargin=*]
		\item\label{item_lem_ce_property1}
		a line segment between two copies of $p\in S$ only intersects
		line segments incident to another copy of $p$,
		\item\label{item_lem_ce_property2}
		all copies of a point are almost on a horizontal line; that is, 
		if $p\in S$ is above (below) $q\in S$ then 
		any line through two different copies of $p$
		is above (below)
		any copy of $q$. 
	\end{enumerate}
	To this end, we place the $m$ copies of a point $p=(x,y)$
	at $p_i = (x+i\varepsilon,y+(i\varepsilon)^2)$ for $i=0,\ldots,m-1$. 
	For sufficiently small $\varepsilon>0$, all points from $S'$ have distinct $x$- and $y$-coordinates and the above conditions are fulfilled.
	
	Let $F'$ be a maximum crossing family in $S'$.
	We will show how to find a crossing family $F$ in $S$ of size at least $|F'|/m$.	
	If $|F'|\le m$, we are clearly done since any segment in $S$ is a crossing family of size~1.
	Hence assume that $F'$ contains more than $m$ segments.
	Then no segment $f' \in F'$ can be incident to two copies of the same point of $S$ due to property~\ref{item_lem_ce_property1}.
	Thus every segment $f' \in F'$ connects (copies of) two distinct points of $S$. 
	
	Let $F$ be the set of segments (without multiplicity) on $S$ induced by $F'$ 
	and by contracting the $m$ copies of $p_i$ to $p$, for each $p \in S$. Formally,
	$F = \{ \overline{pq} \mid \exists i,j\colon \overline{p_iq_j} \in F'\}$;
	Figure~\ref{fig:cross_fam_fmn_lemma_2} gives an illustration.
	Let $G_F$ be the geometric graph induced by $F$,  
	which has the set of end points of $F$ as (drawn) vertices 
	and $F$ as (drawn) edges. More formally, 
	$G_F=(V,E,\phi,\psi)$ with $V=\{p \in S \mid \exists q\colon \overline{pq}\in F \}$, $E=\{ \{p,q\} \mid \overline{pq} \in F\}$, $\phi: V\to \mathbb{R}^2, p \mapsto p$, and $\psi \colon \{p,q\} \mapsto \overline{pq}$ for any $\{p,q\}\in E$.
	By construction $G_F$ is a geometric thrackle,
	and due to Lemma~\ref{lem:even_cycles_thrackle},
	$G_F$ contains no even cycles. 
	Observe that according to property~\ref{item_lem_ce_property2},
	the neighbors of a vertex $p$ in $G_F$ can either be all above $p$ or all below $p$,
	as no segment of $F'$ connected from above to a copy~$p_i$  can cross a segment of $F'$ connected from below to a copy~$p_j$.
	As a consequence, $G_F$ is bipartite and thus contains no odd cycles. Hence, $G_F$ is acyclic; equivalently, a forest.
		
	As long as $G_F$ contains vertices of degree larger than~1, we continue as follows:
	Let $p \in V(G_F)$ be a leaf incident to a vertex $q \in V(G_F)$ of degree larger than~1. 
	We construct $F''$ by removing all segments in $F'$ incident to copies of $q$,
	and by inserting segments connecting all copies of $q$ to copies of $p$ 
	in the way such that all those segments cross;
	Figure~\ref{fig:cross_fam_fmn_lemma_2} gives an illustration of this modification.
	By construction, $|F''| \ge |F'|$ holds and thus $F''$ is another maximal crossing family in $S'$.
	We can replace $F'$ by~$F''$.
	
	\begin{figure}[htb]
		\centering
		\includegraphics{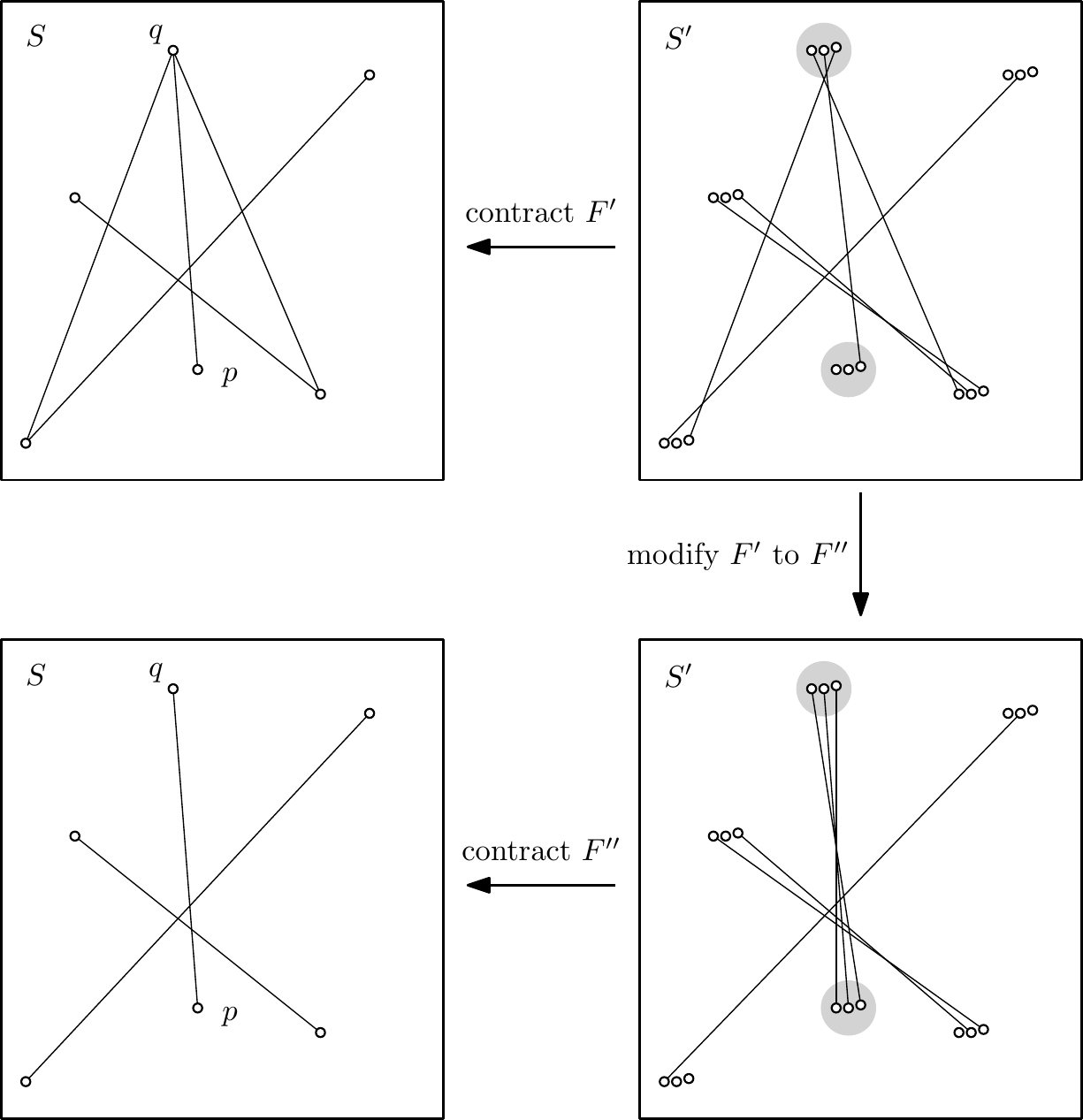}
		\caption{An illustration of the contracting process to attain $F$ from  $F'$, 
			and an illustration of the modification of $F'$ 
			in which all copies of $p$ and $q$ are ``connected'' to each other.}
		\label{fig:cross_fam_fmn_lemma_2}
	\end{figure}
	
	We can iteratively repeat this process, and in every step the number of vertices of degree larger than~1 strictly decreases.
	Therefore we can assume that $G_F$ contains no vertices of degree greater than~1.
	As a consequence, $F$ is a (not necessarily maximal) crossing family in $S$ with $|F| \ge |F'|/m$.
        Therefore, $\cf(S')\le |F'| \le m\cdot |F| \le m\cdot\cf(S) = mk$.
\end{proof}

From the point set depicted in Figure~\ref{fig:n9no3fam}, it follows by Theorem~\ref{thm:factor} that there are sets of $n$ points with no crossing family larger than $2\lceil \frac{n}{9} \rceil$.
This already improves the upper bound $\cf(n)\le\frac{n}{4}$ by Aronov et al.~\cite[Section~6]{aegkkps-94}. Evans and Saeedi~\cite{evans2019problems} constructed a set of 24 points with no crossing family of size 6 or more, which yields the upper bound $\cf(n) \le 5\lceil \frac{n}{24} \rceil$ presented there.

\begin{table}[htb]
\begin{center}
{\small
\begin{tabular}{|c||*{10}{c|}}
\hline
       $k$    & 1 & 2 & 3 & 4 & 5 & 6 & 7 & 8 & 9 & 10 \\  
\hline
\hline
Currently largest sets $S_k$  & 4 & 9 & 14 & 20 & 25 & 29 & 34 & 41 & 45 & 50 \\
\hline
\end{tabular}
}
\end{center}
\caption{The sizes of the largest known point sets $S_k$ with a maximum crossing family of size $k$, that is, with $\cf(S_k)=k$. For $k \le 3$ the sizes are best possible.}
\label{tab:crossfam}
\end{table}

To further improve this bound we searched for sets with small crossing families. 
For $k \ge 5$ we partially extended several smaller sets (for example by doubling the number of points similarly to the process described in the proof of Theorem~\ref{thm:factor}) and used heuristics such as simulated annealing and Brownian motion to optimize them. 
To be more precise,
for our computations we used, among others, the python framework \verb|networkx|  and its \verb|number_of_cliques| implementation to count and minimize the number of $k$-crossing families.

Table~\ref{tab:crossfam} summarizes the sizes of the currently known largest sets with maximum crossing family of size $k\le 10$. 

\begin{figure}%[htb]
	\centering
	\includegraphics{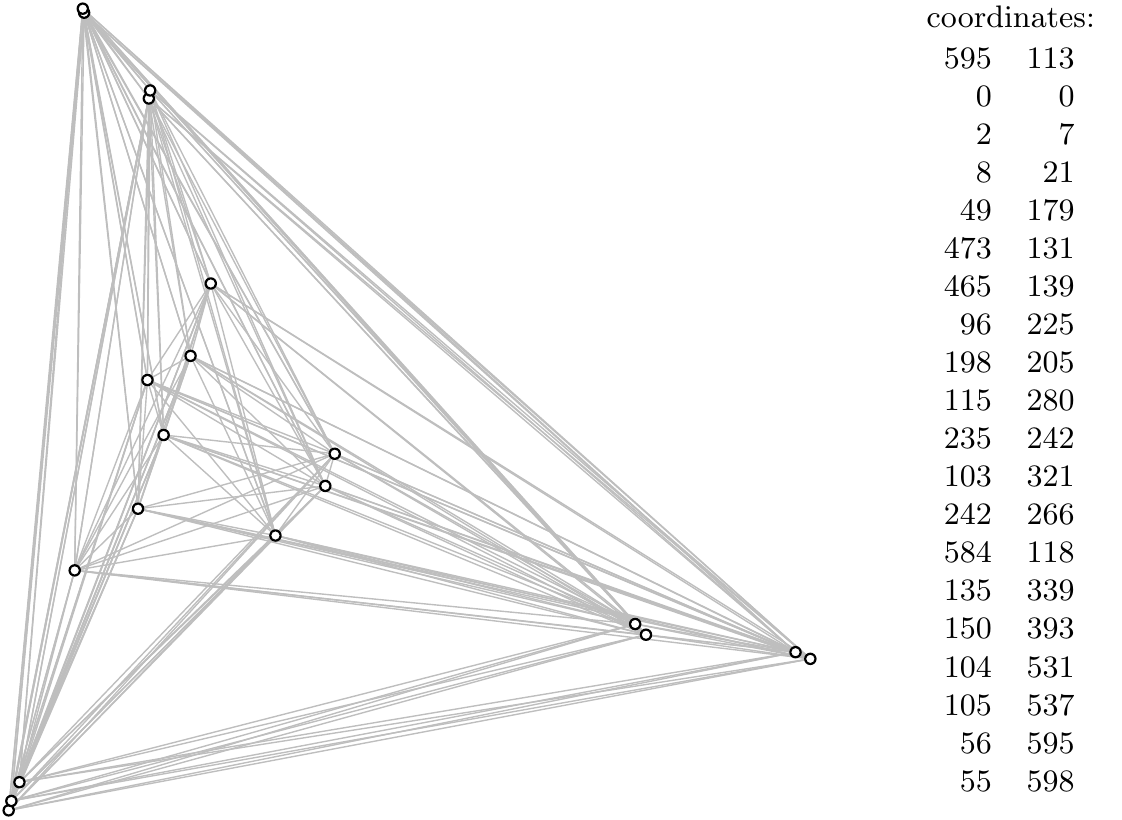}
	\caption{\label{fig:n20no5fam}
		A set $S$ of 20 points with no 5-crossing family.
	}
\end{figure}

\begin{figure}%[htb]
	\centering
	\includegraphics{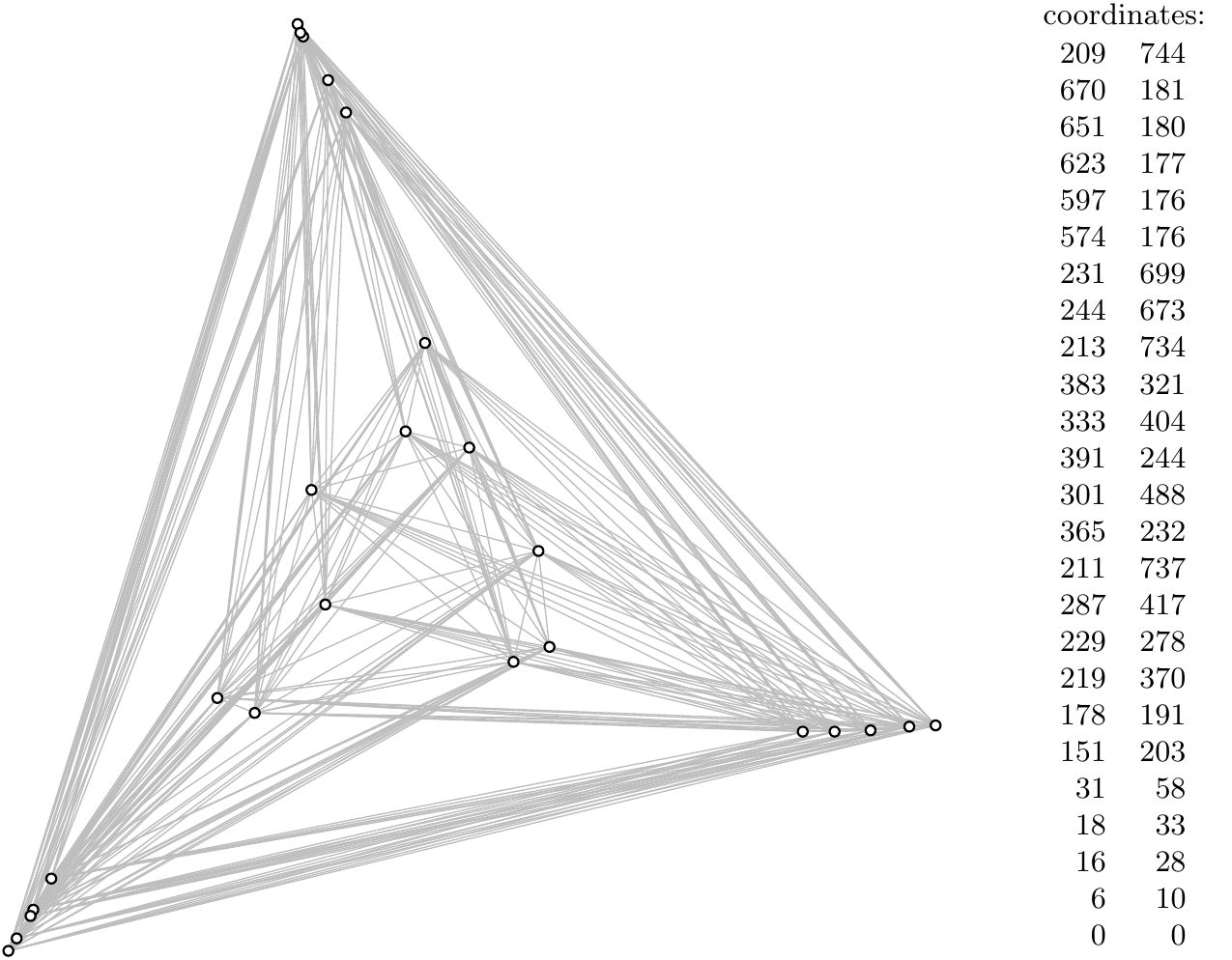}
	\caption{\label{fig:n25no6fam}
		A set $S$ of 25 points with no 6-crossing family.
	}
\end{figure}

Using Theorem~\ref{thm:factor} together with
the sets of 20 points containing no 5-crossing family (see Figure~\ref{fig:n20no5fam}) we get the bound $\cf(n) \le 4\lceil \frac{n}{20} \rceil$.
Also we have a set of 25 points containing no 6-crossing family (see Figure~\ref{fig:n25no6fam}), which implies  $\cf(n) \le 5\lceil \frac{n}{25} \rceil$ and therefore gives a slightly better upper bound for certain values of~$n$.

By doubling the 20-point configuration without 5-crossing families from Figure~\ref{fig:n20no5fam},
we obtained a 40-point configuration without 9-crossing families.
Using heavy computer assistance, we managed to extend this configuration to a 
41-point configuration without 9-crossing families which is shown in Figure~\ref{fig:n41_no9fam}.
By Theorem~\ref{thm:factor} this witnesses $\cf(n) \le 8\lceil \frac{n}{41} \rceil $, where $\frac{8}{41} \approx 0.195$.

\begin{corollary}\label{cor:upperbound}
	It holds that $\cf(n) \le 8\lceil \frac{n}{41} \rceil$.
\end{corollary}

We remark that
even though finding the 41-point configuration without 9-crossing families took hundreds of CPU days on a cluster,
it can be verified within a few minutes.

\begin{figure}%[htb]
	\centering
	\includegraphics{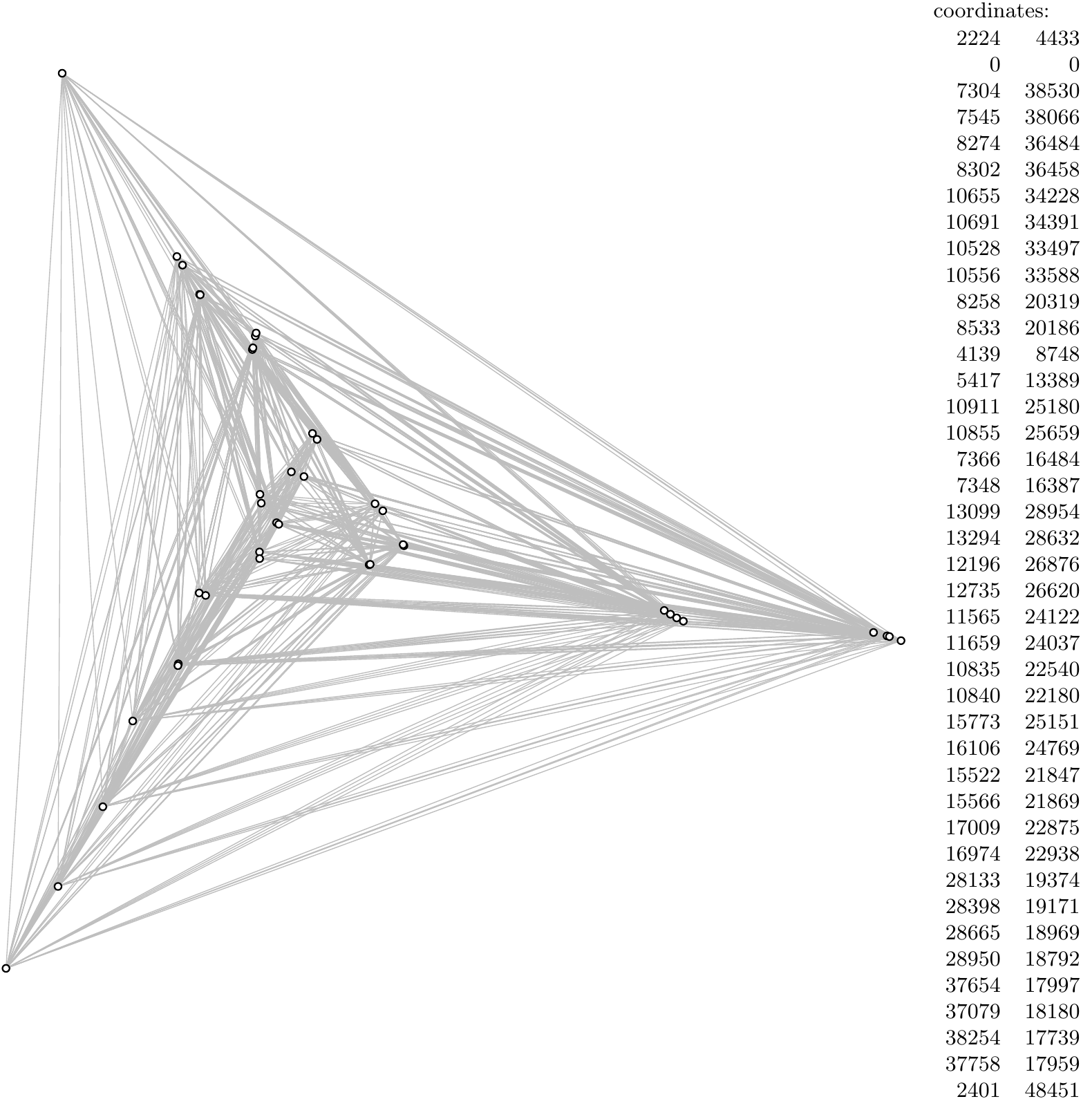}
	\caption{\label{fig:n41_no9fam}
		A set $S$ of 41 points with no 9-crossing family.
	}
\end{figure}

%=================================================================================================

{
	\small
\bibliography{crossingfamily}
\bibliographystyle{alphaabbrv-url}
}

\end{document}